\documentclass[10pt,journal,twocolumn]{IEEEtran}

\usepackage{cite}
\usepackage{amsmath}
\usepackage{amsfonts}
\usepackage{amssymb}
\usepackage{dsfont}
\usepackage[amsmath,thmmarks]{ntheorem}
\usepackage{trfsigns}
\usepackage{multirow}
\usepackage{arydshln}
\usepackage{subfigure}
\usepackage{float}
\usepackage{xcolor}
\usepackage{url}
\usepackage{trfsigns}
\usepackage{multicol}

\usepackage{graphicx}

\DeclareGraphicsExtensions{.pdf} 

\theoremstyle{plain}
\newcommand{\hhnewtheorem}[2]{
  \newtheorem{#2}{#1}
  \newenvironment{#1}[1][CCC]{
    \begin{#2}[##1]
      \nopagebreak
      \begin{adjustwidth}{0.05\textwidth}{0.05\textwidth}}
      {\end{adjustwidth}
    \end{#2}}}
\newcommand{\hhnewproof}[3]{
  \newtheorem{#3}{#2}
  \newenvironment{#1}{
    \begin{#3}
      \nopagebreak
      \begin{adjustwidth}{0.05\textwidth}{0.05\textwidth}}
      {\end{adjustwidth}
    \end{#3}}}
\theoremseparator{.}
\hhnewtheorem{Theorem}{thm}
\hhnewtheorem{Proposition}{prop}
\hhnewtheorem{Claim}{claim}
\hhnewtheorem{Corollary}{cor}
\hhnewtheorem{Lemma}{lemma}
\theorembodyfont{\upshape}
\theoremsymbol{\ensuremath{\diamondsuit}}
\theoremseparator{.}
\hhnewtheorem{Example}{example}
\hhnewtheorem{Remark}{remark}
\hhnewtheorem{Definition}{defn}
\hhnewtheorem{Assumption}{assumption}
\theoremstyle{nonumberbreak}
\theoremseparator{:}
\hhnewtheorem{Algorithm}{algorithm}
\theoremstyle{plain}
\theoremseparator{}
\theoremsymbol{\openbox}
\theoremstyle{nonumberplain}
\hhnewproof{Proof}{Proof:}{proof}

\newcommand{\openbox}{\leavevmode
  \hbox to.77778em{
  \hfil\vrule
  \vbox to.675em{\hrule width.6em\vfil\hrule}
  \vrule\hfil}}

\newcommand{\tr}{\mathop{\rm tr\/}}
\newcommand{\diag}{\mathop{\rm diag\/}}
\newcommand{\E}{\mathsf{E}}
 
\newcommand{\R}{\mathbb{R}}

\newcommand{\trans}{\mathsf{T}}
\newcommand{\vm}[1]{\boldsymbol{#1}}
\newcommand{\Gpdf}{g}
\newcommand{\GpdfN}{\Gpdf_{\N}}
\newcommand{\GpdfM}{\Gpdf_{\M}}
\newcommand{\GMweight}{\omega}

\newcommand{\meanV}{\vm{\mu}}
\newcommand{\CovM}{\vm{\Sigma}}

\newcommand{\CovD}{\vm{D}}
\newcommand{\covD}{d}
\newcommand{\covDV}{\vm{\covD}}
\newcommand{\covDVr}{\covDV^{(\indK,\indL)}}
\newcommand{\EqVar}{\xi^{2}}
\newcommand{\EqStdr}{\xi_{\indK,\indL}}
\newcommand{\EqVarr}{\EqVar_{\indK,\indL}}

\newcommand{\I}{\vm{I}}
\newcommand{\Eeps}{\delta^{2}}

\newcommand{\xvar}{\sigma^{2}_{x}}

\newcommand{\NB}{R}
\newcommand{\nb}{r}

\newcommand{\Blen}{Q}
\newcommand{\blen}{q}
\newcommand{\PDF}{p}
\newcommand{\dx}{\mathrm{d}}
\newcommand{\estMA}{\vm{W}}

\newcommand{\estMrA}{\estMA_{\indK,\indL}}

\newcommand{\MSE}{\mathsf{mse}}
\newcommand{\MSEsu}{\MSE_{\mathsf{eq}}}

\newcommand{\Fb}{\beta}
\newcommand{\nvarT}{\sigma^{2}}

\newcommand{\indK}{k}
\newcommand{\K}{K}
\newcommand{\indL}{l}
\newcommand{\Lk}{L_{\indK}}
\newcommand{\onesV}{\vm{1}}

\newcommand{\mtxA}{\vm{A}}

\newcommand{\nvecT}{\vm{n}}
\newcommand{\nvec}{\tilde{\nvecT}}

\newcommand{\indN}{n}
\newcommand{\M}{M}
\newcommand{\N}{N}
\newcommand{\Nind}{n}
\newcommand{\cwL}{\N}

\newcommand{\PDFpost}{q}
\newcommand{\symbolX}{x}

\newcommand{\vecXT}{\vm{\symbolX}}
\newcommand{\vecXTr}{\vecXT_{\indK,\indL}}
\newcommand{\vecYT}{\vm{y}}
\newcommand{\vecX}{\tilde{\vecXT}}
\newcommand{\vecXr}{\vecX_{\indK,\indL}}
\newcommand{\vecY}{\tilde{\vecYT}}
\newcommand{\slangle}{\langle}
\newcommand{\srangle}{\rangle}
\newcommand{\dlangle}{\langle\langle}
\newcommand{\drangle}{\rangle\rangle}
\newcommand{\vecZT}{\vm{z}}
\newcommand{\vecZTr}{\vecZT_{\indK,\indL}}

\newcommand{\FreeEd}{f}
\newcommand{\NR}{u}
\newcommand{\nr}{a}
\newcommand{\nrb}{b}
\newcommand{\Alambda}{\lambda}
\newcommand{\AlambdaV}{\vm{\Alambda}}
\newcommand{\partF}{Z}
\newcommand{\partFcond}{\partF(\vecYT,\mtxA,\AlambdaV)}

\newcommand{\AXiNR}{\Xi^{(\NR)}_{\N}(\AlambdaV)}

\newcommand{\vecV}{\vm{v}}

\newcommand{\mtxQ}{\vm{\mtxQEle}}
\newcommand{\mtxQr}{\vm{\mtxQEle}_{\indK,\indL}}
\newcommand{\mtxQu}{\mtxQ^{(\NR)}}
\newcommand{\mtxQuS}{\mtxQ^{*}}
\newcommand{\mtxQEle}{Q}

\newcommand{\mtxQtil}{\tilde{\mtxQ}\vphantom{\mtxQ}}
\newcommand{\mtxQutil}{\mtxQtil^{(\NR)}}
\newcommand{\mtxQutilS}{\mtxQtil^{*}}

\newcommand{\funcT}{T}
\newcommand{\funcTu}{\funcT^{(\NR)}}
\newcommand{\funcTur}{\funcTu_{\indK,\indL}}
\newcommand{\funcG}{G}

\newcommand{\funcGu}{\funcG^{(\NR)}}

\newcommand{\AmtxA}{\vm{\Sigma}}
\newcommand{\MGF}{\phi}
\newcommand{\AMGF}{\MGF^{(\NR)}_{\indK,\indL}}
\newcommand{\AMGFN}{\AMGF}
\newcommand{\Qaa}{p}
\newcommand{\Qab}{q}
\newcommand{\Qtilaa}{\tilde{p}}
\newcommand{\Qtilab}{\tilde{q}}
\newcommand{\cMGFuN}{C^{(\NR)}_{\N}}

\newcommand{\funcH}{h}
\newcommand{\funcHN}{\funcH_{\N}}

\newcommand{\nvecsu}{\vm{\eta}}

\begin{document}

\title{Analysis of MMSE Estimation for \\
Compressive Sensing of Block Sparse Signals%
\thanks{A previous version of this paper was published in
IEEE ITW'11, Paraty Brazil, October 16-20, 2011.  There was,
however, a mistake in the main result and proof, albeit the 
main conclusion given in 
Corollary~\ref{cor:mmse_equivalence}
of the paper was still correct.  The results and 
proofs have been corrected in this arXiv version.}
}

\author{\IEEEauthorblockN{Mikko Vehkaper{\"a}, 
Saikat Chatterjee, and Mikael Skoglund}       

\IEEEauthorblockA{
School of Electrical Engineering and the ACCESS Linnaeus Center \\
KTH Royal Institute of Technology, SE-10044, Stockholm, Sweden \\
E-mail: \url{{mikkov, sach, skoglund}@kth.se}}}

\maketitle

\begin{abstract}
\emph{Minimum mean square error} (MMSE) estimation of 
\emph{block sparse} signals 
from noisy linear measurements is considered.  
Unlike in the standard \emph{compressive sensing} setup
where the non-zero entries of the signal are
independently and uniformly distributed across the vector 
of interest, the information bearing components appear here in 
large mutually dependent \emph{clusters}.
Using the \emph{replica method} from statistical physics,
we derive a simple closed-form solution for the 
MMSE obtained by the optimum estimator.
We show that the MMSE is a version
of the \emph{Tse-Hanly formula} with 
system load and MSE scaled by parameters that depend on the 
\emph{sparsity pattern} of the source.
It turns out that 
this is equal to the MSE obtained by a \emph{genie-aided} MMSE estimator 
which is informed in advance about the exact locations of the non-zero 
blocks.  The asymptotic results obtained by the non-rigorous 
replica method are found to have an excellent agreement
with finite sized numerical simulations.

\end{abstract}

\section{Introduction}

\emph{Compressive sensing} (CS) 
\cite{Candes-Romberg-Tao-2006, Donoho-2006} tackles the 
problem of recovering a high-dimensional \emph{sparse} vector
from a set of linear measurements.  
Typically the number of observations
is much less than the number of elements in the vector 
of interest, making naive reconstruction attempts inefficient.
In addition to being an under-determined problem, 
the measurements may suffer from additive noise.
Under such conditions, the signal model for the 
noisy CS measurements can be written as
\begin{equation}
\vspace*{-0.1ex}
\vecYT = \mtxA \vecXT + \nvecT \in \R^{\M},
\vspace*{-0.1ex}
\label{eq:sysmodel_w_A}
\end{equation}
where $\vecXT\in\R^{\N}$ is the sparse vector of interest,
$\mtxA\in\R^{\M \times \N}$ the measurement
matrix, and $\nvecT\in\R^{\M}$ 
represents the measurement noise.
By assumption, $\M < \N$ and 
only some of the elements of $\vecXT$  are non-zero.
The task of CS is then to infer $\vecXT$, given 
$\mtxA$, $\vecYT$ and possibly some information about the 
sparsity of $\vecXT$ and the statistics of the 
noise $\nvecT$.  

In this paper, the vector $\vecXT$ is assumed 
to have a special \emph{block sparse} structure.
Such sparsity patterns have recently been found, e.g., 
in multiband signals and 
multipath communication channels (see, e.g.,
\cite{Stojnic-2010, Stojnic-Parvaresh-Hassibi-2009,
Eldar-Kuppinger-Bolcskei-2010, Baraniuk-2010} 
and references therein).
More precisely, the source is considered to be 
$\K$ \emph{block sparse} so that for any realization of $\vecXT$, its 
information bearing entries occur in \emph{at most} $\K$ 
non-overlapping \emph{clusters}.
This is markedly different from the conventional 
sparsity assumption in CS, where the individual 
non-zero components 
are independently and uniformly distributed over $\vecXT$.
Given the $\K$ block sparse source, we study the
\emph{minimum mean square error} (MMSE)
estimation of $\vecXT$, assuming full knowledge 
of the statistics of the input $\vecXT$ and the noise $\nvecT$.
Albeit this is an optimistic scenario for 
practical CS problems, it provides a lower bound on the MSE 
for any other reconstruction method.  
Also, knowing the benefits of having the 
statistics of the system at the estimator 
gives a hint how much the sub-optimum blind schemes could improve if 
they were to learn the parameters of the problem.

The main result of the paper is the 
\emph{closed-form MMSE} for the CS of block sparse signals.  
The solution turns out to be of a particularly simple form, namely,
the \emph{Tse-Hanly formula} \cite{Tse-Hanly-1999} 
where the system load and MSE are scaled by  parameters
that depend on the \emph{sparsity pattern} of the source.
This is found to be equal to the MSE obtained by a genie-aided
MMSE estimator that is informed in advance 
the locations of the non-zero blocks.
The result implies that if the statistics of the 
block sparse CS problem are known, the MMSE
is independent of the knowledge about 
the positions of the non-vanishing blocks.

Finally, we remark that the analysis in the paper are obtained via
the \emph{replica method} (RM) from statistical physics.
Albeit the RM is non-rigorous, it has
been used with great success for the \emph{large system} analysis of, e.g.,
multi-antenna systems
\cite{Moustakas-Simon-2003,Muller-2003},
code division multiple access 
\cite{Tanaka-2002, Guo-Verdu-2005}, vector precoding \cite{MGM-2008},
iterative receivers \cite{vehkapera-thesis-short-2010} and 
compressed sensing
\cite{Rangan-Fletcher-Goyal-isit-2009,
Guo-Baron-Shamai-allerton2009,
Kabashima-Wadayama-Tanaka-isit2010,
Tanaka-Raymond-2010}.
The main difference here compared to 
\cite{Tanaka-2002, Guo-Verdu-2005, Muller-2003, MGM-2008,
vehkapera-thesis-short-2010, 
Rangan-Fletcher-Goyal-isit-2009,
Guo-Baron-Shamai-allerton2009,
Kabashima-Wadayama-Tanaka-isit2010,
Tanaka-Raymond-2010} is that
the elements of the $\K$ block sparse vector $\vecXT$
are \emph{neither independent nor identically distributed}. 
This requires a slight modification to the standard replica treatment,
as detailed in the Appendix.

\subsection{Notation}

The probability density function (PDF)  of 
a random vector (RV) $\vecXT \in \R^{\cwL}$ 
(assuming one exists) is written 
as $\PDF(\vecXT)$ and conditional densities are denoted 
$\PDF(\vecXT \mid \cdots)$. 
The related PDFs \emph{postulated} by the estimator are denoted 
$\PDFpost(\vecX)$ and $\PDFpost(\vecX \mid \cdots)$, respectively. 
For further discussion on the
so-called \emph{generalized posterior mean estimation}
using true and postulated probabilities, see for example,
\cite{Tanaka-2002, Guo-Verdu-2005, 
vehkapera-thesis-short-2010}.
We denote
$\vecYT\sim\PDF(\vecYT)=\GpdfN(\vecYT \mid \meanV;\,\CovM)$
for a RV $\vecYT$ that is drawn according to
the $\cwL$-dimensional Gaussian density 
$\GpdfN(\vecYT \mid \meanV;\,\CovM)$
with mean $\meanV\in\R^{\cwL}$ and
covariance $\CovM\in\R^{\cwL\times\cwL}$.  
For a vector $\vecXT$ that is drawn according to a
\emph{Gaussian mixture} density, we have
\begin{equation}
\vecXT\sim\PDF(\vecXT) =
\sum_{r=1}^{R}\GMweight_{r}\GpdfN(\vecXT \mid \meanV_{r};\,\CovM_{r}),
\label{eq:GMpdf}
\end{equation}
where the density parameters $\GMweight_{r}$ satisfy
$\sum_{r=1}^{R}\GMweight_{r} = 1$ and
$\GMweight_{r}\geq 0$ for all $r=1,\ldots,R.$

We write $\onesV_{\Blen}\in\R^{\Blen}$
for the all-ones vector having $\Blen$ elements, 
and given vectors 
$\{\vm{d}_{\nb}\in\R^{\N}\}_{\nb=1}^{\NB}$, 
the diagonal matrix
$\vm{D}=\diag(\vm{d}_{1},\ldots,\vm{d}_{\NB})\in\R^{\N\NB\times\N\NB}$ 
has vector
$[\vm{d}_{1}^{\trans} \; \cdots \;\,\vm{d}_{\NB}^{\trans}]^{\trans}
\in\R^{\N\NB}$
on the main diagonal and zeros elsewhere. 
Superscript $^{\trans}$ denotes the transpose of a matrix.

\section{System Model}

\subsection{Block Sparsity}

\label{sec:block_sparsity}

Let $\cwL = \Blen\NB$, where $\Blen,\NB$ are positive integers, be the 
length of the sparse vector $\vecXT$. Furthermore, let $\vecXT$
be composed of $\NB$ equal length 
sub-vectors $\{\vecXT_{\nb}\}_{\nb=1}^{\NB}$, that is,
\begin{equation}
\vecXT = 
\begin{bmatrix}
\vecXT_{1} \\
\vdots \\
\vecXT_{\NB}
\end{bmatrix} \in \R^{\cwL},
\;\text{ where }\;
\vecXT_{\nb} = 
\begin{bmatrix}
\symbolX_{\nb, 1} \\
\vdots \\
\symbolX_{\nb, \Blen}
\end{bmatrix} \in \R^{\Blen}.
\label{eq:x_block_vector}
\end{equation}
Instead of considering \emph{strict block sparsity}
where some of the sub-vectors $\{\vecXT_{\nb}\}_{\nb=1}^{\NB}$
are \emph{exactly} equal to zero
\cite{Stojnic-2010, Stojnic-Parvaresh-Hassibi-2009,
Eldar-Kuppinger-Bolcskei-2010, 
Baraniuk-2010}, we let
$\vecXT$ be drawn from 
the Gaussian mixture density
\begin{equation}
\PDF(\vecXT)
= \sum_{\indK=1}^{\K}
\sum_{\indL=1}^{\Lk}
\GMweight_{\indK,\indL}\GpdfN(\vecXT \mid \vm{0};\,\CovD_{\indK,\indL}),
\label{eq:pdf_of_x}
\end{equation}
where $\K$ is an integer,
\begin{IEEEeqnarray}{l}
\Lk = \binom{\NB}{\indK}, \quad
\GMweight_{\indK} = \sum_{\indL=1}^{\Lk}\GMweight_{\indK,\indL}
\quad \text{ and } \quad
 \sum_{\indK=1}^{\K}
 \GMweight_{\indK} = 1.
 \IEEEeqnarraynumspace
\label{eq:GMparameters}
\end{IEEEeqnarray}
Here $\GMweight_{\indK}$ denotes the probability of observing 
$\indK$ information bearing blocks in a realization of $\vecXT$, 
$\Lk$ the number of combinations how such $\indK$ blocks 
can be arranged in $\vecXT$, and 
$\GMweight_{\indK,\indL}\geq 0$ the probabilities related to these
arrangements.
In the following, the input of
\eqref{eq:sysmodel_w_A} in the
event that $\vecXT$ is drawn according to the $(\indK,\indL)$th mixture
density is denoted $\vecXTr$.

To impose a block sparse structure on the vector of interest,
the diagonal covariance matrices in \eqref{eq:pdf_of_x}
\begin{IEEEeqnarray}{rCl}
\CovD_{\indK,\indL} &=& \diag(\covDVr_{1},\ldots,\covDVr_{\NB}),
\qquad \indL = 1,\ldots,\Lk, 
\IEEEeqnarraynumspace
\label{eq:Dr_matrix_defn}
\end{IEEEeqnarray}
are required to be distinct for all 
$(\indK,\indL) \neq (\indK',\indL')$ and taking
only two values on the block diagonals 
\begin{equation}
\covDVr_{\nb} = 
 \onesV_{\Blen}\Eeps \; \text{ and }\;
\covDVr_{\nb} =  \onesV_{\Blen}\xvar,
 \quad \forall \nb,\indK,\indL.
 \end{equation}
Here $\Eeps$ and $\xvar$ represent the expected
signal powers of the sparsity inducing and
information bearing components, respectively.
With these assumptions, the per-component variance
\begin{equation}
\frac{1}{\cwL}\E\|\vecXTr\| =
\frac{1}{\cwL}\tr(\CovD_{\indK,\indL})
= \frac{\indK}{\NB} \xvar + \frac{\NB - \indK}{\NB} \Eeps,
\label{k_sparse_Dkl}
\end{equation}
is independent of the permutation index $\indL = 1,\ldots,\Lk$.
\begin{defn}
	\label{defn:approx_K_block_sparse}
Let $\xvar = 1$ and assume that $\vecXT$ is a RV with 
density \eqref{eq:pdf_of_x}.  We say $\vecXT$ is an
\emph{approximately} $\K$ block sparse RV if
$\Eeps \ll \xvar$ and all  matrices
$\{\CovD_{\indK,\indL}\}$ 
satisfy \eqref{eq:Dr_matrix_defn}~--~\eqref{k_sparse_Dkl}.
If $\Eeps\to 0^{+}$, we simply say that $\vecXT$ is 
a $\K$ block sparse RV.
\end{defn}

\subsection{Posterior Mean Estimation}
\label{eq:pme_w_A}

Let $\vecXT$ be drawn according to 
\eqref{eq:pdf_of_x} and assume that we observe the noisy measurements
$\vecYT$ in \eqref{eq:sysmodel_w_A}.
We assume that the noise 
$\nvecT\sim\GpdfM(\nvecT\mid\vm{0};\, \nvarT\I_{\M})$
is Gaussian and independent of the signal $\vecXT$ and 
the measurement matrix $\mtxA$.
Furthermore, we let $\mtxA$ be independent of $\vecXT$ 
with independent identically distributed
(IID) Gaussian entries with zero mean and variance $1/\M$.

Given the above assumptions, let
\begin{equation}
\vecY = \mtxA \vecX + \nvec \in \R^{\M},
\label{eq:post_sysmodel_w_A}
\end{equation}
be the signal model postulated by the estimator.
We assign the postulated densities
$\PDFpost(\vecX)$ and 
$\PDFpost(\nvec)$ 
to the input and noise vectors, respectively.
In the following, $\PDFpost(\vecY \mid \mtxA,\vecX)=
\PDFpost(\vecY=\vecYT \mid \mtxA,\vecX)$ means 
that the realizations of the postulated measurement 
vectors $\vecY$ match the outputs $\vecYT$ 
of the true signal model \eqref{eq:sysmodel_w_A}, but it can be 
that the input $\PDFpost(\vecX)\neq \PDF(\vecXT)$ and noise
$\PDFpost(\nvec) \neq \PDF(\nvecT)$ statistics are mismatched.
If we define an expectation operator 
\begin{equation}
\dlangle \cdots \drangle_{\PDFpost}
\triangleq
\int \cdots \, \PDFpost(\vecX \mid \vecY, \vm{A}) 
\dx \vecX 
= \int
\cdots
\frac{\PDFpost(\vecY \mid \mtxA, \vecX ) \PDFpost(\vecX)}
{\PDFpost(\vecY\mid \mtxA)} \dx \vecX,
\label{eq:mmse_joint_est_w_A}
\end{equation}
a (mismatched) MMSE estimate of $\vecXT$ for 
the linear model \eqref{eq:sysmodel_w_A}, given 
$\vecYT, \mtxA, \PDFpost(\vecX)$ and 
$\PDFpost(\vecY \mid \mtxA,\vecX)$, is then simply
$\dlangle \vecX \drangle_{\PDFpost}$.
%

\begin{lemma}
The MMSE estimate of $\vecXT$ for the  signal model 
\eqref{eq:sysmodel_w_A} reads
$\dlangle \vecX \drangle_{\PDF}$, i.e., $\PDFpost = \PDF$ for 
all densities. 
\label{defn:mmse-optimum}
\end{lemma}
\begin{proof}
The result follows by simple algebra and is omitted due to space constraints.
\end{proof}

\begin{prop}
\label{thm:mmse_estimator}
The MMSE estimate of the 
$\K$ block sparse signal $\vecXT$ with density \eqref{eq:pdf_of_x},
given noisy measurements \eqref{eq:sysmodel_w_A}, reads
\begin{equation}
\dlangle \vecX \drangle_{\PDF} =
\sum_{\indK=1}^{\K}
\sum_{\indL=1}^{\Lk}
\frac{\GMweight_{\indK,\indL}\PDF_{\indK,\indL}(\vecYT\mid\mtxA)}
{\PDF(\vecYT\mid\mtxA)} \estMrA\, \vecYT,
\label{eq:optimum_mmse_w_A}
\end{equation}
where
\begin{IEEEeqnarray}{rCl}
\PDF_{\indK,\indL}(\vecYT\mid\mtxA) &\triangleq& 
\GpdfM(\vecYT \mid \vm{0};\, 
\mtxA \CovD_{\indK,\indL} \mtxA^{\trans}+\nvarT\I_{\M}), \\
\PDF(\vecYT\mid\mtxA) &=& 
\sum_{\indK=1}^{\K}
\sum_{\indL=1}^{\Lk}
\GMweight_{\indK,\indL}\PDF_{\indK,\indL}(\vecYT\mid\mtxA),
\end{IEEEeqnarray}
and
\begin{equation}
\estMrA = \CovD_{\indK,\indL} \mtxA^{\trans}
\left(\mtxA \CovD_{\indK,\indL} \mtxA^{\trans}
+ \nvarT\I_{\M} \right)^{-1}.
\end{equation}

\end{prop}

\begin{proof}
	Let $\vm{A}\in\R^{\N\times\N}$ be a symmetric positive 
	definite matrix.  Then, using Gaussian integrals
	\begin{IEEEeqnarray}{rCl}
	\label{eq:Gint1}
	 e^{\frac{1}{2}\vm{b}^{\trans}\vm{A}^{-1}\vm{b}} &=& 
	 \sqrt{\frac{\det(\vm{A})}{2 \pi^{N}}}\int 
	 e^{-\frac{1}{2}\vecXT^{\trans}\vm{A}\vecXT 
	 + \vm{b}^{\trans} \vecXT}
	\dx \vecXT,  \\
	e^{\frac{1}{2}\vm{b}^{\trans}\vm{A}^{-1}\vm{b}} \vm{A}^{-1}\vm{b} 
	&=& \sqrt{\frac{\det(\vm{A})}{2 \pi^{N}}}
	\int \vecXT e^{-\frac{1}{2}\vecXT^{\trans}\vm{A}\vecXT + \vm{b}^{\trans} \vecXT}
	\dx \vecXT, 	\IEEEeqnarraynumspace
	\label{eq:Gint2}
	\end{IEEEeqnarray}
	where $\vm{x},\vm{b}\in\R^{\N}$ on 
		Lemma~\ref{defn:mmse-optimum}, along with the identities
	\begin{IEEEeqnarray}{l}
	\I - \vm{U}
	\left(\vm{B}^{-1} + \vm{V}\vm{U}
	\right)^{-1} \vm{V}
	=\left(\I + \vm{U}\vm{B}\vm{V}
	\right)^{-1}, \\
	\det\big(\vm{C}^{-1} + 
	\vm{U}\vm{V}^{\trans} \big) = 
	\det\big(\I + \vm{V}^{\trans}\vm{C}^{-1}\vm{U}
	\big) \det\big(\vm{C}^{-1}\big),
	\IEEEeqnarraynumspace
	\end{IEEEeqnarray}
	where all matrices are assumed to be of proper 
	size and $\vm{B},\vm{C}$ invertible,  yields 
	Proposition~\ref{thm:mmse_estimator}.
\end{proof}

Given the MMSE estimate of Proposition~\ref{thm:mmse_estimator},
we are now interested in computing the per-component MSE 
\begin{equation}
\MSE(\nvarT) = 
\E\| 
\vecXT - \dlangle \vecX \drangle_{\PDF} \|^{2} / \N,
\label{eq:mse_defn}
\end{equation}
when the dimensions of  
$\mtxA$ grow large with 
fixed ratio  $\Fb = \N / \M$, and the number of 
blocks $\NB = \cwL / \Blen$ stays finite.
The desired result is obtained in two steps: 
1) the replica method is used in
Sec.~\ref{sec:eq_awgn_channel}
to show that the original 
problem can be transformed to a set of simpler ones in the 
large system limit; 2) the solutions to the transformed  
problems are given in Sec.~\ref{sec:eq_awgn_channel_mmse}.

\section{Main Results}

\subsection{Equivalent AWGN Channel}
\label{sec:eq_awgn_channel}

Let the indexes 
$\indK = 1,\ldots,\K$ and $\indL = 1,\ldots,\Lk$
be as in the previous section.
Define a set of additive white Gaussian noise (AWGN) channels
for all $\indK$ and $\indL$
\begin{equation}
\vecZTr = \vecXTr + \EqStdr\vm{\nvecsu} \in \R^{\cwL}, \quad 
\vm{\nvecsu} \sim \GpdfN(\vm{\nvecsu} \mid \vm{0};\, \I_{\cwL}),
\label{eq:awgn_channel}
\end{equation}
where $\vecXTr$ is a zero-mean Gaussian RV 
with covariance $\CovD_{\indK,\indL}$ and $\EqStdr > 0$.
Let the events of observing the $(\indK,\indL)$th channel
\eqref{eq:awgn_channel} be independent 
and occur with probability $\GMweight_{\indK,\indL}$ for all 
$\indK$ and $\indL$.
Furthermore, let
\begin{IEEEeqnarray}{rCl}
\slangle \cdots \srangle_{\PDFpost}^{(\indK,\indL)}
&\triangleq&
\int \cdots \, \PDFpost(\vecXr \mid \vecZTr) 
\dx \vecXr  \IEEEnonumber\\
&=&
\int
\cdots
\frac{\PDFpost(\vecZTr \mid \vecXr ) \PDFpost(\vecXr)}
{\PDFpost(\vecZTr)} \dx \vecXr,
\label{eq:mmse_joint_est_w_su}
\end{IEEEeqnarray}
be an expectation operator
similar to \eqref{eq:mmse_joint_est_w_A}, but related to 
the $(\indK,\indL)$th AWGN channel \eqref{eq:awgn_channel}.
The MMSE estimate of 
$\vecXTr$ given the channel outputs $\vecZTr$ is then given by
\begin{equation}
\slangle \vecXr \srangle_{\PDF}^{(\indK,\indL)} =
\CovD_{\indK,\indL} 
\left(\CovD_{\indK,\indL} + \EqVarr\I_{\N}  \right)^{-1}
\vecZTr.
\label{eq:mmse_output_vector}
\end{equation}
We denote the per-component 
MSE of the estimates 
$\slangle \vecXr \srangle_{\PDF}^{(\indK,\indL)}$ 
\begin{equation}
\MSEsu^{(\indK,\indL)}(\EqVarr) 
= \E \big\{\|\vecXTr\|^{2} 
- \E \|\slangle \vecXr \srangle^{(\indK,\indL)}_{\PDF}\|^{2}\big\} / \N,
\label{eq:mse_su_defn}
\end{equation}
where the expectation is w.r.t the joint distribution of all 
variables associated with \eqref{eq:awgn_channel}.
The per-component MSE averaged over 
the realizations of the channels \eqref{eq:awgn_channel}
is thus 
\begin{equation}
\MSEsu(\{\EqVarr\})
= \sum_{\indK=1}^{\K} \sum_{\indL=1}^{\Lk}
		\GMweight_{\indK,\indL}
		\MSEsu^{(\indK,\indL)}(\EqVarr).
\label{eq:mse_su_avg_defn}		
\end{equation}

\begin{claim}
\label{claim:decoup_mmse}
	Let $\vm{x}$
	be (approximately) $\K$ block sparse as given in Definition~\ref{defn:approx_K_block_sparse}.
	In the large system limit when $\M, \N, \Blen \to \infty$ with finite 
	and fixed ratios  $\Fb = \N / \M$ and $\NB = \N / \Blen$,
	\begin{equation}
		\MSE(\nvarT) \to 
		\MSEsu(\{\EqVarr\}),
	\end{equation}
	where $\MSE(\nvarT)$ is given in \eqref{eq:mse_defn}
	and $\MSEsu(\{\EqVarr\})$ in \eqref{eq:mse_su_avg_defn}.
	The noise variance $\EqVarr> 0$, on the other hand, 
	is the solution to the fixed 
	point equation
	\begin{IEEEeqnarray}{rCl}
		\EqVarr &=& \nvarT + \Fb\, \MSEsu^{(\indK,\indL)}(\EqVarr),
		\label{eq:fp-su-xi}
	\end{IEEEeqnarray}
	for all $\indK$ and $\indL$.
\end{claim}
\begin{proof}
The proof is based on the \emph{replica method} (see, e.g., 
\cite{Tanaka-2002, Guo-Verdu-2005, Muller-2003, Moustakas-Simon-2003, 
MGM-2008, vehkapera-thesis-short-2010,
Rangan-Fletcher-Goyal-isit-2009,
Guo-Baron-Shamai-allerton2009,
Kabashima-Wadayama-Tanaka-isit2010,
Tanaka-Raymond-2010} for similar results in
communication theory and signal processing) 
from statistical physics.  
The main difference to the standard approach 
is that here the elements of the input vector $\vecXT$ are 
\emph{neither independent nor identically distributed}.
Thus, the \emph{decoupling result} proved 
for the CDMA systems \cite{Guo-Verdu-2005} 
cannot be straightforwardly extended to our case.  
Alternative derivation is sketched below and in part in the Appendix.
%

To start, let us define a modified \emph{partition function} related to 
the posterior mean estimator \eqref{eq:mmse_joint_est_w_A} as
\cite{Merhav-isit-2010} 
\begin{equation}
\partFcond = 
\int
\e^{\AlambdaV^{\trans} \vecX}
\PDFpost(\vecYT\mid\mtxA, \vecX)
\PDFpost(\vecX) \dx \vecX,
\label{eq:partf_cond_1}
\end{equation}
where $\AlambdaV\in\R^{\N}$ is a constant vector.  
The posterior mean estimator of $\vecXT$ given 
in \eqref{eq:mmse_joint_est_w_A}
can then be written
as the \emph{gradient} with respect to (w.r.t) $\AlambdaV$
at $\AlambdaV = \vm{0}$
of the \emph{free energy}, i.e.,  
\begin{IEEEeqnarray}{rCl}
\label{eq:nabla_est}
\dlangle \vecX \drangle_{\PDFpost} 
&=& \nabla_{\AlambdaV} \log \partFcond\Big|_{\AlambdaV=\vm{0}}.
\end{IEEEeqnarray}
Similarly, if \eqref{eq:nabla_est} is the 
optimum MMSE estimate, that is
$\PDFpost = \PDF$, 
and we denote the \emph{free energy density} 
\begin{equation}
\FreeEd(\vecYT,\mtxA,\AlambdaV) = 
\frac{1}{\N} 
\log \partFcond,
\label{eq:FreeEd}
\end{equation}
the \emph{average per-component MMSE} is given by
\begin{equation}
\MSE 
= \tr \Big(\E \big\{ 
\nabla_{\AlambdaV\AlambdaV}^{2} \,
\FreeEd(\vecYT,\mtxA,\AlambdaV)\big|_{\AlambdaV=\vm{0}}
\big\}\Big),
\label{eq:nabla_cov}
\end{equation}
where the expectations are w.r.t. the joint density of
$(\vecYT,\mtxA)$.  
Unlike in \cite{Merhav-isit-2010}, however,
direct computation of the free energy (density) 
is not possible here.  
We thus resort to the non-rigorous RM to calculate 
\eqref{eq:FreeEd} and then use the relation \eqref{eq:nabla_cov}
to obtain the final result. 
The details are given in the Appendix.
\end{proof}

\subsection{Performance of MMSE Estimation of Block Sparse Signals}
\label{sec:eq_awgn_channel_mmse}

Claim~\ref{claim:decoup_mmse} asserts that the MSE of the 
estimator \eqref{eq:optimum_mmse_w_A} in the original setting 
\eqref{eq:sysmodel_w_A} can be obtained in the large system limit from \eqref{eq:mse_su_avg_defn}.  Given the Claim~\ref{claim:decoup_mmse}
holds, a bit of algebra gives the following proposition.

\begin{prop}
	\label{prop:mse}
	Let $\xvar = 1$,
	$\EqVarr>0$ and $\N\to\infty$.  Then 
	the per-component MSE \eqref{eq:mse_su_defn} 
	is given by
	\begin{IEEEeqnarray}{l}
	\label{eq:K-largeN-mse}
		\MSEsu^{(\indK,\indL)}(\EqVarr) 
		 = 
	\frac{\indK}{\NB} \frac{\EqVarr}{\EqVarr+1} 
	+ \frac{\NB - \indK}{\NB} \frac{\Eeps\EqVarr}{\EqVarr+\Eeps},
	\end{IEEEeqnarray}
	where $\EqVarr$ is the solution of 
	\eqref{eq:fp-su-xi}.
	When the source is strictly block sparse, that is $\Eeps\to 0^{+}$, 
	\begin{equation}
	\EqVar_{\indK} \overset{\Eeps\to 0^{+}}{=}
	\frac{1}{2} \big(-1 + \Fb_{\indK} + \nvarT + \sqrt{
	   4 \nvarT + (-1 + \Fb_{\indK} + \nvarT)^2}\,\big),
	   	\label{eq:EqVar_closedform}
	\end{equation}
	where $\EqVar_{\indK}=\EqVar_{\indK,\indL} \;\forall \indL=1,\ldots,\Lk$
	and we denoted $\Fb_{\indK} = \frac{\indK}{\NB} \Fb$
	for notational convenience.  Thus, given Claim~\ref{claim:decoup_mmse}
	holds, the MMSE of the block sparse system 
	is given by
	\begin{equation}
	\label{eq:MMSE_block_sparse}
	\MSE(\nvarT) =
			\sum_{\indK=1}^{\K}
			\GMweight_{\indK}
			\frac{\indK}{\NB} \frac{\EqVar_{\indK}}{\EqVar_{\indK}+1},
		\end{equation}
	in the large system limit.
\end{prop}

\begin{remark}
	Note that the MSE is independent of the 
	distribution 
	$\{\GMweight_{\indK,\indL}\}_{\indL}^{\Lk}$
	that makes up 	$\GMweight_{\indK}$ (see \eqref{eq:GMparameters}).
\end{remark}
\begin{remark}
	As $\Eeps\to 0^{+}$, the noise variance 
	\eqref{eq:EqVar_closedform} becomes the 
	\emph{Tse-Hanly solution}
	\cite{Tse-Hanly-1999}
	for equal power users but
	\emph{with a modified user load} 
	$\Fb_{\indK} = \frac{\indK}{\NB} \Fb$. 
	The same noise variance is obtained by a genie-aided MMSE receiver,
	conditioned on the event that $\vecXT$ is sampled 
	from one of the $\Lk$ mixtures indexed by $\indK$.
	The MSE \eqref{eq:MMSE_block_sparse}, on the other hand, 
	is a summation of
	the related MSEs but weighted with the probability of 
	having $\indK$ non-zero blocks in a
	realization of the source vector $\vecXT$.
	Thus, there is no loss in not knowing the 
	positions of the zero blocks in advance 
	if we use the optimum
	MMSE receiver for very large $\K$ block sparse systems.  
	Note that the equivalent AWGN channel model in
	Sec.~\ref{sec:eq_awgn_channel} already implies this point.
	For practical settings with
	finite sized sensing matrices, however, 
	this does not strictly hold.	
\end{remark}

\begin{cor}
	The MMSE estimator
	\eqref{eq:optimum_mmse_w_A} has the same MSE
	in the large system limit 
	as a genie-aided MMSE estimator that
	knows in advance the positions of the non-zero
	blocks in $\vecXT$.
	\label{cor:mmse_equivalence}
\end{cor}

To empirically verify the analytical results, 
we have plotted in Fig.~\ref{F:plot} the MSE of  
estimator \eqref{eq:optimum_mmse_w_A}, obtained 
via computer simulations. 
The theoretical MSE given in Proposition~\ref{prop:mse} is 
given as well.  
In all simulation cases we have set $\GMweight_{\indK,\indL} =
\GMweight = 1 / \sum_{\indK=1}^{\K} \Lk$ so that
$\GMweight_{\indK}  = \GMweight\Lk$.
For the selected cases the theory matches 
Monte Carlo simulations very well.
\begin{figure}[t]
\centering
\includegraphics[width=1\columnwidth]{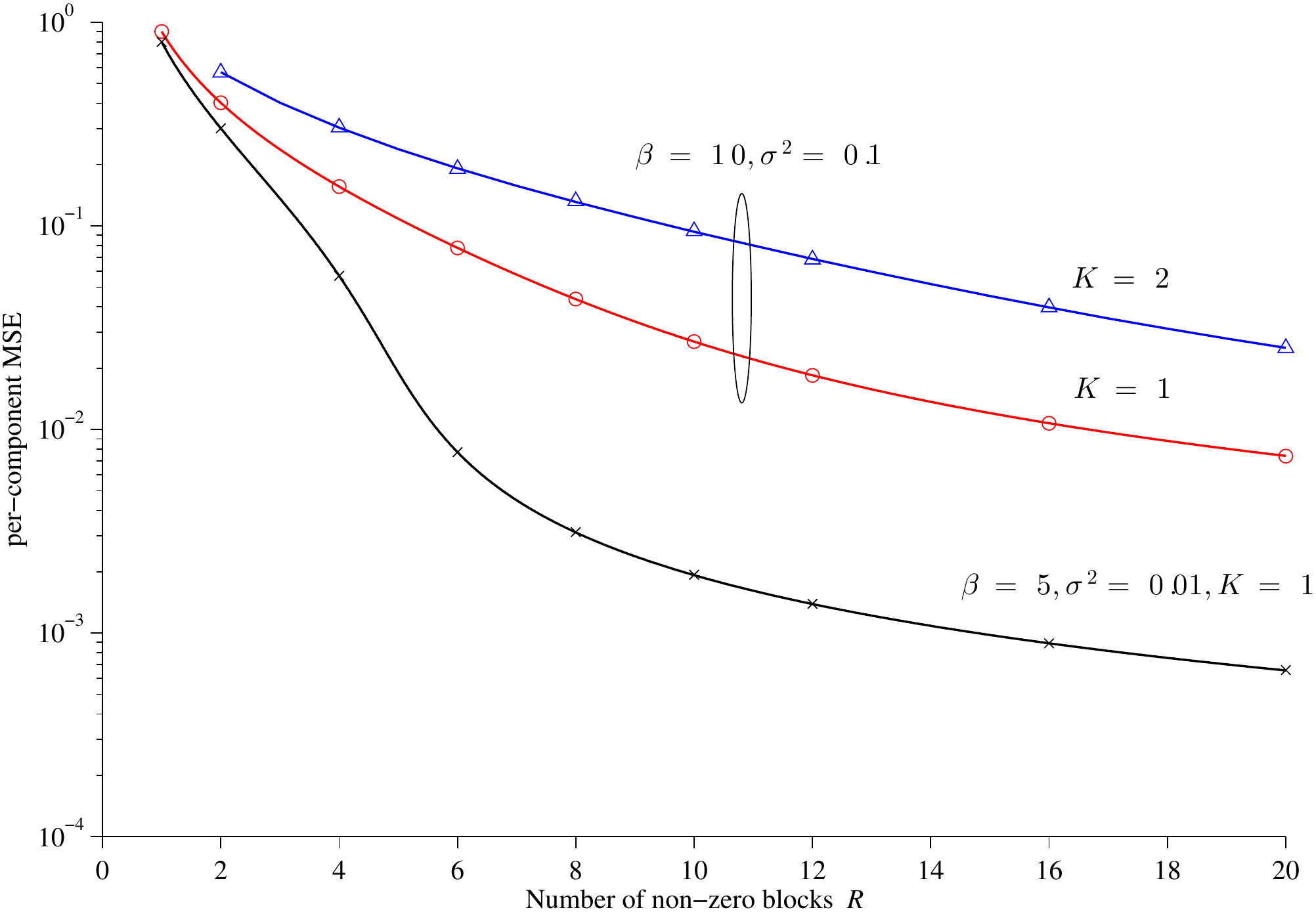}
\caption{Per-component MMSE for the CS of block sparse signal. 
Solid lines are obtained from Proposition~\ref{prop:mse} 
with $\Eeps\to 0^{+}$ and markers 
depict numerical Monte Carlo simulations where 
$\xvar=1$, $\Eeps = 10^{-6}$ and  $\N=1200$.}
\label{F:plot}
\end{figure}

\section{Conclusions}

Minimum mean square error estimation of 
block sparse signals 
from noisy linear measurements was considered. 
The main result of the paper is the
closed-form MMSE for the CS of such signals.  
The solution turned out to be of a particularly simple form, namely,
the Tse-Hanly formula with a scaling by parameters 
that depend on the sparsity pattern of the source.
The result implies that if the statistics of the 
block sparse CS problem are known, the MMSE
is independent of the knowledge about 
the positions of the non-vanishing blocks.

\appendix[Derivation of Claim~\ref{claim:decoup_mmse}]

Let us assume that the free energy density
\eqref{eq:FreeEd}
is self-averaging w.r.t. the quenched randomness $(\mtxA,\vecYT)$
in the large system limit $\N\to\infty$.  
Then \eqref{eq:FreeEd} can be written as
\begin{IEEEeqnarray}{rCl}
\FreeEd 
&=& \lim_{\N\to\infty} 
\frac{1}{\N} 
\lim_{\NR\to 0} \frac{\partial}{\partial \NR}
\log \E_{\vecYT,\mtxA} \{ \partFcond^{\NR}\},
\IEEEeqnarraynumspace
\end{IEEEeqnarray}
where $\NR$ is a real parameter.
The \emph{replica trick} consists of 
treating $\NR$ as an integer while calculating the 
expectations, but taking the limit as if 
$\NR$ was real valued outside the expectation.  
The second step is to exchange the limits and
write the power of $\NR$ inside the expectation using the set 
$\{\vecXT_{[\nr]}\}_{\nr=1}^{\NR}$ of
replicated random vectors, resulting to,
\begin{IEEEeqnarray}{rCl}
  \FreeEd_{\mathsf{rm}}
&=& 
\lim_{\NR\to 0} \frac{\partial}{\partial \NR}
\lim_{\N\to\infty} 
\frac{1}{\N}
\log \AXiNR,
\label{eq:FreeEd_2}
\end{IEEEeqnarray}
where 
$\vecXT_{[\nr]}$ are IID with density
$\PDF(\vecXT)$ and
\begin{equation}
\AXiNR = 
\E_{\vecYT, \mtxA}
\left\{
\int 
\prod_{\nr=1}^{\NR}
\PDF(\vecXT_{[\nr]}) 
\e^{\AlambdaV^{\trans} \vecXT_{[\nr]}}
\PDF(\vecYT\mid\mtxA, \vecXT_{[\nr]})
\dx \vecXT_{[\nr]}\right\}.
\label{eq:AXi_1}
\end{equation}
Unfortunately, these steps are non-rigorous and there is no 
general proof yet
under which conditions $\FreeEd_{\mathsf{rm}}$ equals $\FreeEd$.  
For more discussion and details, see, e.g.,
\cite{Tanaka-2002, Guo-Verdu-2005, Muller-2003, 
Moustakas-Simon-2003, MGM-2008, vehkapera-thesis-short-2010}.

Let $\vecXT_{[0]}$ be the true vector of interest, 
independent of $\{\vecXT_{[\nr]}\}_{\nr=1}^{\NR}$
and distributed as $\vecXT$.   Plugging
$\vecYT = \mtxA \vecXT_{[0]} + \nvecT$ 
to \eqref{eq:AXi_1}, the
average over the additive noise vector
$\nvecT\sim\GpdfM(\nvecT \mid \vm{0};\, \nvarT\I_{\M})$
can be assessed using \eqref{eq:Gint1}. 
Furthermore, recalling that the true and postulated source vectors have
GM densities \eqref{eq:pdf_of_x}, 
we obtain \eqref{eq:AXi_1_2} at the top of the next page,
where $\E^{(\indK,\indL)}_{\{\vecXT_{[\nr]}\}}$ denotes 
expectation over the vectors 
$\vecXT_{[\nr]} \sim \GpdfN(\vecXT_{[\nr]}\mid \vm{0};\,\CovD_{\indK,\indL}), \nr = 0, 1, \ldots,\NR$.
\begin{figure*}[!t]
\normalsize
\begin{IEEEeqnarray}{rCl}
\AXiNR &=& 
\bigg[\frac{(\NR+1)^{-1}}{(2 \pi \nvarT)^{\NR}}\bigg]^{\frac{\M}{2}}
\sum_{\indK=1}^{\K}
\sum_{\indL=1}^{\Lk}
\GMweight_{\indK,\indL} 
\E^{(\indK,\indL)}_{\{\vecXT_{[\nr]}\}}
\Bigg\{ 
\E_{\mtxA}\Bigg\{
\exp
\bigg(
-\frac{1}{2\nvarT (\NR+1)}
\bigg[\sum_{\nr=0}^{\NR}
\NR\| \mtxA \vecXT_{[\nr]}\|^{2}
- 
\sum_{\nr=0}^{\NR}
\sum_{\nrb \neq \nr}
(\mtxA\vecXT_{[\nr]})^{\trans}
(\mtxA\vecXT_{[\nrb]})
\bigg]
\bigg)
\Bigg\} \Bigg\} 
\IEEEnonumber\\
\label{eq:AXi_1_2}
\vspace*{-1em}
\end{IEEEeqnarray}
\hrulefill
\vspace*{-1em}
\end{figure*}
Next, let
\begin{equation}
\vecV=[(\vecV_{0})^{\trans}  \, \cdots\;
(\vecV_{\NR})^{\trans}]^{\trans}\in\R^{\M(\NR+1)},
\end{equation}
be a RV composed of $\NR+1$ sub-vectors
$\vecV_{\nr} = \Fb^{-1/2} \mtxA \vecXT_{[\nr]} \in\R^{\M}$.
Also denote $\mtxQr = \mtxQu_{\indK,\indL} \otimes \I_{M}$ where
$\mtxQu_{\indK,\indL} \in \R^{(\NR+1)\times(\NR+1)}$
and the $(\nr,\nrb)$th element of $\mtxQu_{\indK,\indL}$ 
is given by
$\mtxQEle_{\indK,\indL}^{[\nr,\nrb]} = \vecXT_{[\nr]}^{\trans} 
\vecXT^{}_{[\nrb]} / \N$ where 
$\vecXT_{[\nr]}, \vecXT_{[\nrb ]}\sim 
\GpdfN(\vecXT\mid \vm{0};\,\CovD_{\indK,\indL})
$ for all $\nr,\nrb = 0,1,\ldots,\NR$.
Then, \eqref{eq:AXi_1_2} can be written in the form
\begin{IEEEeqnarray}{rCl}
\AXiNR &=&
\bigg[\frac{(2 \pi \nvarT)^{-\NR}}{\NR+1}\bigg]^{\frac{\M}{2}}
\sum_{\indK=1}^{\K}
\sum_{\indL=1}^{\Lk}
\GMweight_{\indK,\indL} \IEEEnonumber\\
&& \times\E^{(\indK,\indL)}_{\{\vecXT_{[\nr]}\}}
\Big\{ 
\E_{\vecV\sim\GpdfM(\vecV\mid \vm{0};\, \mtxQr)}
\Big\{
\e^{- \frac{1}{2}
\vecV^{\trans}\AmtxA \vecV
}
\Big\} \Big\}, \IEEEeqnarraynumspace
\label{eq:AXi_1_4}
\end{IEEEeqnarray}
where 
$\AmtxA = 
(\Fb/\nvarT)
 \left[\I_{\NR} - \onesV_{\NR}\onesV_{\NR}^{\trans} / (1+\NR)
 \right]\in \R^{(\NR+1)\times(\NR+1)}$.
Using \eqref{eq:Gint1} to integrate over the Gaussian RV $\vecV$ in 
\eqref{eq:AXi_1_4} yields
\begin{IEEEeqnarray}{l}
\AXiNR \IEEEnonumber\\
\;= 
\sum_{\indK=1}^{\K} \sum_{\indL=1}^{\Lk}
\GMweight_{\indK,\indL} 
\E^{(\indK,\indL)}_{\{\vecXT_{[\nr]}\}}
\left\{
\e^{\N \Fb^{-1} \funcGu(\mtxQu_{\indK,\indL})}
\prod_{\nr=1}^{\NR}\e^{\AlambdaV^{\trans} \vecXT_{[\nr]}}
\right\}, \IEEEeqnarraynumspace
\label{eq:AXi_3}
\end{IEEEeqnarray}
where 
\begin{IEEEeqnarray}{l}
\e^{\funcGu(\mtxQu_{\indK,\indL})}
=
\sqrt{
\frac{(2 \pi\nvarT)^{-\NR}}{(1+\NR)
\det(\I + \AmtxA\mtxQu_{\indK,\indL})}}.
\IEEEeqnarraynumspace
\end{IEEEeqnarray}

To compute the expectations w.r.t.,
$\{\vecXT_{[\nr]}\}_{\nr=0}^{\NR}$ 
for all $\indK = 1,\ldots,\K, \indL = 1,\ldots,\Lk$
in \eqref{eq:AXi_3}, 
we write the measure of 
the matrix $\mtxQu_{\indK,\indL}$ as
\begin{IEEEeqnarray}{l}
  \mu_{\N} (\mtxQu_{\indK,\indL}) 
  \IEEEnonumber\\
  \; =
\E^{(\indK,\indL)}_{\{\vecXT_{[\nr]}\}}\Bigg\{
\prod_{\nr=1}^{\NR}
\e^{\AlambdaV^{\trans}\vecXT_{[\nr]}} 
   \prod_{\nr\leq\nrb}\delta\left( 
  \vecXT_{[\nr]}^{\trans} \vecXT^{}_{[\nrb]}
  = \N \mtxQEle_{\indK,\indL}^{[\nr,\nrb]}
  \right)\Bigg\}, \IEEEeqnarraynumspace
  \label{eq:meas_Q_dirac}
\end{IEEEeqnarray}
and integrate w.r.t. \eqref{eq:meas_Q_dirac}.
Writing the Dirac measures in \eqref{eq:meas_Q_dirac} in terms of
(inverse) Laplace transform 
and invoking saddle point integration
(see \cite[Appendix~A]{MGM-2008} for details), 
we get
\begin{IEEEeqnarray}{l}
\AXiNR
= 
\sum_{\indK=1}^{\K} \sum_{\indL=1}^{\Lk}
\GMweight_{\indK,\indL} \e^{\N\funcTur(\AlambdaV)},
\label{eq:AXi_4}
\end{IEEEeqnarray}
where $\mtxQutil_{\indK,\indL}\in\R^{(\NR+1)\times(\NR+1)}$ 
is a symmetric matrix.  
To obtain \eqref{eq:AXi_4}, we defined an auxiliary function
\begin{IEEEeqnarray}{l}
\funcTur(\AlambdaV) = 
\sup_{\mtxQu_{\indK,\indL}}\Big\{
\Fb^{-1} \funcGu(\mtxQu_{\indK,\indL})  \IEEEnonumber\\
\;  - \inf_{\mtxQutil_{\indK,\indL}}\Big\{
\tr(\mtxQu_{\indK,\indL} \mtxQutil_{\indK,\indL}) - 
\lim_{\N\to\infty}\N^{-1} 
\log \AMGFN (\mtxQutil_{\indK,\indL}, \AlambdaV; \N)
\Big\}
\bigg\}, \IEEEnonumber\\
\label{eq:optfunc_1}
\end{IEEEeqnarray}
where
\begin{IEEEeqnarray}{l}
\AMGFN (\mtxQutil_{\indK,\indL}, \AlambdaV; \N)  \IEEEnonumber\\
\quad = \E^{(\indK,\indL)}_{\{\vecXT_{[\nr]}\}} \left\{
\e^{\AlambdaV^{\trans}\vecXT_{[\nr]}} 
\e^{
\tr\left[
\mtxQutil_{\indK,\indL} \sum_{\indN=1}^{\N} \vecXT^{(\NR)}_{\Nind}(\vecXT^{(\NR)}_{\Nind})^{\trans} 
\right]}
\right\},  \IEEEeqnarraynumspace
	\label{eq:mod_mgf}
\end{IEEEeqnarray}
denotes the moment generating function (MGF) of \eqref{eq:meas_Q_dirac}.
We also wrote 
$\vecXT^{(\NR)}_{\Nind} = [\symbolX_{[0],\indN} \, \cdots\; \symbolX_{[\NR],\indN}]\in\R^{\NR+1}$, 
where $\Nind = (\nb-1)\Blen+\blen$ for 
$\nb = 1,\ldots,\NB$ and $\blen = 1,\ldots,\Blen$ in the notation of 
\eqref{eq:x_block_vector}.

To make the  optimization problems in \eqref{eq:optfunc_1} tractable, we assume that their solutions 
are the \emph{replica symmetric} (RS) matrices
(see, e.g., \cite{Tanaka-2002, Guo-Verdu-2005, Muller-2003, 
Moustakas-Simon-2003, MGM-2008, vehkapera-thesis-short-2010} on discussion
about this assumption)
\begin{IEEEeqnarray}{rCl}
	\label{eq:Qstar1}
	\mtxQuS_{\indK,\indL} &=&
	(\Qaa_{\indK,\indL}-\Qab_{\indK,\indL})\I_{\NR+1} + \Qab_{\indK,\indL}\onesV_{\NR+1}^{}\onesV_{\NR+1}^{\trans},
	\IEEEeqnarraynumspace \\
	\mtxQutilS_{\indK,\indL} &=&
	(\Qtilaa_{\indK,\indL}-\Qtilab_{\indK,\indL})\I_{\NR+1} +
	\Qtilab_{\indK,\indL}\onesV_{\NR+1}^{}\onesV_{\NR+1}^{\trans},
	\IEEEeqnarraynumspace
\label{eq:Qstar2}
\end{IEEEeqnarray}
respectively, where $\Qaa_{\indK,\indL},
\Qab_{\indK,\indL},\Qtilaa_{\indK,\indL}, \Qtilab_{\indK,\indL}$ 
are real parameters.
Under the RS assumption, we get the simplifications
\begin{IEEEeqnarray}{l}
\tr(\mtxQu_{\indK,\indL} \mtxQutil_{\indK,\indL})
=
(\NR+1)(\Qaa_{\indK,\indL}\Qtilaa_{\indK,\indL} + 
\NR\Qab_{\indK,\indL}\Qtilab_{\indK,\indL}) \\
\quad \xrightarrow{\NR\to 0} 0,
\end{IEEEeqnarray}
and
\begin{IEEEeqnarray}{rCl}
\funcGu(\Qaa_{\indK,\indL},\Qab_{\indK,\indL})
&=& 
-\frac{\NR}{2} \log 
[\nvarT + \Fb (\Qaa_{\indK,\indL} 
	- \Qab_{\indK,\indL})] \IEEEnonumber\\
&&-\frac{\NR}{2} \log (2 \pi\nvarT)
-\frac{1}{2} \log (\NR + 1) \\
\xrightarrow{\NR\to 0} 0. &&
\end{IEEEeqnarray}
From the first extremum in \eqref{eq:optfunc_1} one obtains
$\Qtilaa_{\indK,\indL} = 0$ and
\begin{equation}
	\Qtilab_{\indK,\indL} = \left[\nvarT + \Fb (\Qaa_{\indK,\indL} 
	- \Qab_{\indK,\indL})\right]^{-1},
\end{equation}
where $\Qaa_{\indK,\indL}$ and $\Qab_{\indK,\indL}$ are left as arbitrary but fixed parameters for now.
To proceed with the second optimization problem in \eqref{eq:optfunc_1},
we need to evaluate the MGF \eqref{eq:mod_mgf} under the RS assumption.

Using \eqref{eq:Gint1} right-to-left in \eqref{eq:mod_mgf},
using the RS assumption \eqref{eq:Qstar1}~--~\eqref{eq:Qstar2}
and recalling that the replicas $\{\vecXT_{[\nr]}\}_{\nr=0}^{\NR}$ are 
IID yields after some algebra
\begin{IEEEeqnarray}{l}
\AMGFN (\Qtilab_{\indK,\indL}, \AlambdaV; \N) 
= \cMGFuN(\Qtilab_{\indK,\indL})
\int \E^{(\indK,\indL)}_{\vecXT}\left\{ 
\GpdfN(\vecZT_{\indK,\indL} \mid \vecXT;\,
\Qtilab_{\indK,\indL}^{-1}\I_{\N}) \right\}
 \IEEEnonumber\\
\qquad\times
\left[\E^{(\indK,\indL)}_{\vecX}\left\{\e^{\AlambdaV^{\trans}\vecX}
\GpdfN(\vecZT_{\indK,\indL} \mid \vecX;\,
 \Qtilab_{\indK,\indL}^{-1}\I_{\N}) \right\}\right]^{\NR}
\dx \vecZT_{\indK,\indL}, \IEEEeqnarraynumspace
	\label{eq:mod_mgf_2}
\end{IEEEeqnarray}
where the expectations 
$\E^{(\indK,\indL)}$ are w.r.t.
zero-mean Gaussian RVs with covariance 
$\CovD_{\indK,\indL}$.
The normalization factor
$\cMGFuN(\Qtilab_{\indK,\indL}) 
= \big[(1+\NR) (2 \pi / \Qtilab_{\indK,\indL})^{\NR}\big]^{\N/2}$ 
is due the introduction of the Gaussian densities in 
\eqref{eq:mod_mgf_2}.  Since 
  \begin{equation}
\AMGFN (\Qtilab_{\indK,\indL}, \AlambdaV; \N)
\xrightarrow{\NR\to 0} 1,
\label{eq:MGF_at_u_zero}
\end{equation}
the second optimization in \eqref{eq:optfunc_1}
reduces to the conditions
\begin{IEEEeqnarray}{rCl}
\label{eq:sp-cond_1}
\Qaa_{\indK,\indL} = \lim_{\N\to\infty}
\N^{-1}\,
\E^{(\indK,\indL)}
\left\{\|\vecXT\|^{2} \right\}, \IEEEeqnarraynumspace\\
\Qab_{\indK,\indL} = \lim_{\N\to\infty}
\N^{-1}
\E^{(\indK,\indL)}
\left\{
\|\langle\vecX\rangle^{(\indK,\indL)}_{\PDF}\|^{2}
\right\}, \IEEEeqnarraynumspace
\label{eq:sp-cond_2}
\end{IEEEeqnarray}
when $\AlambdaV = \vm{0}$ and $\NR\to 0$.  
The expectations in 
\eqref{eq:sp-cond_1}~and~\eqref{eq:sp-cond_2}
are w.r.t 
$\vecZT_{\indK,\indL}$, and the independent zero-mean
Gaussian RVs  $\vecXT,\vecX$ with covariance 
$\CovD_{\indK,\indL}$ as in \eqref{eq:mod_mgf_2}.
We also write 
\begin{equation}
	\PDF_{\N}(\vecZT_{\indK,\indL}) = 
	\E^{(\indK,\indL)}_{\vecX}\{
	\GpdfN(\vecZT_{\indK,\indL} \mid
	\vecX;\,\Qtilab_{\indK,\indL}^{-1}\I_{\N})\},
\end{equation}
so that
\begin{IEEEeqnarray}{rCl}
\label{eq:su_mmse_estimator}
\langle\vecX\rangle^{(\indK,\indL)}_{\PDF} &=&
	\frac{1}{\PDF(\vecZT_{\indK,\indL})}
	\E^{(\indK,\indL)}_{\vecX} \big\{
	\vecX \, 
	\GpdfN(\vecZT_{\indK,\indL} \mid \vecX;\,\Qtilab_{\indK,\indL}^{-1}\I_{\N}) 
	\big\}. \IEEEeqnarraynumspace	
\end{IEEEeqnarray}
Note that \eqref{eq:su_mmse_estimator} is the MMSE 
estimator of the Gaussian channel
\begin{equation}
\vecZT_{\indK,\indL} = \vecXT + \vm{\nvecsu}_{\indK,\indL} 
\in \R^{\cwL},  
\label{eq:awgn_channel_appendix}
\end{equation}
when the receiver knows the correct distributions of
$\vm{\nvecsu}_{\indK,\indL} \sim 
\GpdfN(\vm{\nvecsu} \mid \vm{0};\, \Qtilab_{\indK,\indL}^{-1}\I_{\cwL})
$ and
$\vecXT
\sim 
\GpdfN(\vecXT \mid \vm{0};\, \CovD_{\indK,\indL})
$.
Furthermore, from \eqref{eq:sp-cond_1} and \eqref{eq:sp-cond_2}
we get 
\begin{IEEEeqnarray}{rCl}
 \MSEsu^{(\indK,\indL)}
 (\Qtilab_{\indK,\indL}) &=& 	\lim_{\N\to\infty} \frac{1}{\N} \big[
	\E^{(\indK,\indL)}\{\|\vecXT\|^2\} 
	- \E^{(\indK,\indL)}\{\|\slangle \vecX \srangle^{(\indK,\indL)}_{\PDF}\|^2
	\}\big] 
	\IEEEnonumber\\
	&=& \Qaa_{\indK,\indL} - \Qab_{\indK,\indL},
	\label{eq:mmse_of_eq_channel}
\end{IEEEeqnarray}
where $\MSEsu^{(\indK,\indL)}$ is the MMSE of the Gaussian
channel \eqref{eq:awgn_channel_appendix}.

The free energy density under the RS assumption reads 
\begin{equation}
 \FreeEd_{\mathsf{rm-rs}}
 = 
 \lim_{\NR\to 0} 
 \frac{\partial}{\partial \NR}
\lim_{\N\to\infty} 
\frac{1}{\N}
\log \Bigg(
\sum_{\indK=1}^{\K} \sum_{\indL=1}^{\Lk}
\GMweight_{\indK,\indL} \e^{\N\funcTur(\AlambdaV)}
\Bigg).
\end{equation}
Switching the order of the limits once more yields
\begin{IEEEeqnarray}{rCl}
 \FreeEd_{\mathsf{rm-rs}} &=&
\lim_{\N\to\infty} 
\frac{1}{\N}
\lim_{\NR\to 0} \Bigg\{
\bigg(\sum_{\indK=1}^{\K} \sum_{\indL=1}^{\Lk}
\GMweight_{\indK,\indL} \e^{\N\funcTur(\AlambdaV)}\bigg)^{-1}
\IEEEeqnarraynumspace\IEEEnonumber\\
&&\qquad\quad \times\bigg[\sum_{\indK=1}^{\K} \sum_{\indL=1}^{\Lk}
\GMweight_{\indK,\indL} 
\bigg(
\frac{\partial}{\partial \NR}
\e^{\N\funcTur(\AlambdaV)}\bigg)\bigg]
\Bigg\}.
\end{IEEEeqnarray}
Since
\begin{equation}
\funcTur(\AlambdaV)
\xrightarrow{\NR\to 0} 0,
\label{optfunc_at_u_zero}
\end{equation}
by \eqref{eq:GMparameters} the denominator becomes just 
unity and can be omitted.
For the latter part,
\begin{equation}
\frac{\partial}{\partial \NR} \e^{\N\funcTur(\AlambdaV)}
 = 
 \N \e^{\N\funcTur(\AlambdaV)}
 \bigg(
 \frac{\partial}{\partial \NR} \funcTur(\AlambdaV)\bigg),
\end{equation}
where the derivative is assessed
\begin{IEEEeqnarray}{l}
\lim_{\NR\to 0}\frac{\partial}{\partial \NR} \funcTur(\AlambdaV)
= -\frac{1}{2\Fb}\log[\nvarT + \Fb (\Qaa_{\indK,\indL} 
	- \Qab_{\indK,\indL})]
	- \Qab_{\indK,\indL}\Qtilab_{\indK,\indL}
 \IEEEnonumber\\
\qquad+ \frac{1}{2} \left[1 + \log (2\pi/\Qtilab_{\indK,\indL}) \right] 
-\frac{1}{2\Fb}\log(2 \pi \e \nvarT)
\IEEEnonumber\\ 
\qquad + \lim_{\N\to\infty} \frac{1}{\N} \int 
\PDF_{\N}(\vecZT_{\indK,\indL})
\funcHN(\vecZT_{\indK,\indL},\AlambdaV,\Qtilab_{\indK,\indL})
\dx \vecZT_{\indK,\indL}.
\end{IEEEeqnarray}
Note that we defined above the function 
\begin{IEEEeqnarray}{l}
\funcHN(\vecZT_{\indK,\indL},\AlambdaV,\Qtilab_{\indK,\indL}) 
\IEEEnonumber\\
\qquad \qquad \!=
\log\big(\E^{(\indK,\indL)}_{\vecX}\big\{\e^{\AlambdaV^{\trans}\vecX}
\GpdfN(\vecZT_{\indK,\indL} \mid \vecX;\,
 \Qtilab_{\indK,\indL}^{-1}\I_{\N}) \big\}\big),
\IEEEeqnarraynumspace
\end{IEEEeqnarray}
for notational convenience.
Recalling \eqref{optfunc_at_u_zero}, we finally have
the RS free energy density
\begin{IEEEeqnarray}{l}
\FreeEd_{\mathsf{rm-rs}}
= \frac{1}{2} \big[1 - \Fb^{-1}
\log(2 \pi \e \nvarT) 
\big] \IEEEnonumber\\
+ \frac{1}{2}\sum_{\indK=1}^{\K} \sum_{\indL=1}^{\Lk}
\GMweight_{\indK,\indL} 
\big[
\log (2\pi/\Qtilab_{\indK,\indL})
+\Fb^{-1}\log\Qtilab_{\indK,\indL}
-2\Qab_{\indK,\indL}\Qtilab_{\indK,\indL} 
  \big]
\IEEEnonumber\\
+\!
\lim_{\N\to\infty} \frac{1}{\N} 
\sum_{\indK=1}^{\K} \sum_{\indL=1}^{\Lk}
\GMweight_{\indK,\indL} \!
\int \!
\PDF_{\N}(\vecZT_{\indK,\indL})
\funcHN(\vecZT_{\indK,\indL},\AlambdaV,\Qtilab_{\indK,\indL})
\dx \vecZT_{\indK,\indL}, \IEEEeqnarraynumspace
\label{eq:D_log_T_2}
\end{IEEEeqnarray}
where only the last term depends on $\AlambdaV$ and is relevant
for the assessment of the MSE, as given in \eqref{eq:nabla_cov}.

The final task is to compute
$\nabla^{2}_{\AlambdaV\AlambdaV} \funcHN(\vecZT,\AlambdaV,\EqVar)
\big|_{\AlambdaV=\vm{0}}$.
First,
\begin{equation}
\nabla_{\AlambdaV}
\funcHN(\vecZT_{\indK,\indL},\AlambdaV,\Qtilab_{\indK,\indL})
=\frac{
\E^{(\indK,\indL)}_{\vecX}\big\{
\vecX \e^{\AlambdaV^{\trans}\vecX}
\GpdfN(\vecZT_{\indK,\indL} \mid \vecX;\,
 \Qtilab_{\indK,\indL}^{-1}\I_{\N}) \big\}
}
{\E^{(\indK,\indL)}_{\vecX}\big\{\e^{\AlambdaV^{\trans}\vecX}
\GpdfN(\vecZT_{\indK,\indL} \mid \vecX;\,
 \Qtilab_{\indK,\indL}^{-1}\I_{\N}) \big\}},
\end{equation}
so that the estimator \eqref{eq:su_mmse_estimator}
can also be written as
\begin{equation}
\langle\vecX\rangle^{(\indK,\indL)}_{\PDF}
= 
\nabla_{\AlambdaV}
\funcHN(\vecZT_{\indK,\indL},\AlambdaV,
\Qtilab_{\indK,\indL})\big|_{\AlambdaV=\vm{0}}.
\end{equation}
Proceeding similarly, after a bit of algebra we obtain
the conditional covariance matrix of the error 
\begin{IEEEeqnarray}{rCl}
	\vm{E}^{(\indK,\indL)}_{\N}(\vecZT_{\indK,\indL}) &=& 
	\nabla^{2}_{\AlambdaV} \funcHN(\vecZT_{\indK,\indL},\AlambdaV,
	\Qtilab_{\indK,\indL})\Big|_{\AlambdaV=\vm{0}} \IEEEnonumber\\
	&=& \slangle\vecX \vecX^{\trans}\srangle^{(\indK,\indL)}_{\PDF}
	- \slangle \vecX \srangle^{(\indK,\indL)}_{\PDF}
	\big[\slangle \vecX \srangle^{(\indK,\indL)}_{\PDF}
	\big]^{\trans},
\end{IEEEeqnarray}
which is also the error covariance of the estimator 
\eqref{eq:su_mmse_estimator}.
Thus, by \eqref{eq:nabla_cov}, the per-component MSE 
of the original MMSE estimator given in 
Proposition~\ref{thm:mmse_estimator} reads 
\begin{IEEEeqnarray}{l}
\MSE  \IEEEnonumber\\
\;= 
\sum_{\indK=1}^{\K} \sum_{\indL=1}^{\Lk}
\GMweight_{\indK,\indL} \!
\lim_{\N\to\infty} \frac{1}{\N} 
\tr\Bigg(
\int \!
\PDF_{\N}(\vecZT_{\indK,\indL})
\vm{E}^{(\indK,\indL)}_{\N}(\vecZT_{\indK,\indL})
\dx \vecZT_{\indK,\indL}
\Bigg), \IEEEnonumber\\
\end{IEEEeqnarray}
which can be written due to \eqref{eq:sp-cond_1}, \eqref{eq:sp-cond_2}
and \eqref{eq:mmse_of_eq_channel} as
\begin{equation}
\MSE = 
\sum_{\indK=1}^{\K} \sum_{\indL=1}^{\Lk}
\GMweight_{\indK,\indL} 
\MSEsu^{(\indK,\indL)} (\Qtilab_{\indK,\indL}),
\end{equation}
completing the proof.

\bibliography{./jour_short,./conf_short,./references_short}

\bibliographystyle{IEEEtran}

\end{document}